\documentclass[a4paper]{llncs}

\usepackage[utf8x]{inputenc}
\usepackage{amssymb}
\usepackage{amsmath}
\usepackage{amsfonts}
\usepackage{graphicx}
\usepackage{llncsdoc}
\usepackage[english]{babel}
\usepackage[ruled,vlined,linesnumbered]{algorithm2e}
\usepackage{algorithmic}
\usepackage{float}
\usepackage{todonotes}
\usepackage{subfigure}
\usepackage{enumitem}
\usepackage{shuffle}
\usepackage{multirow}
\usepackage{esvect}

\usepackage{tikz}
\usetikzlibrary{matrix, calc, arrows, shapes.callouts}

\usepackage{xspace}

\newcommand{\Pclassbase}{{\sf\bf P}}
\newcommand{\Pclass}{\Pclassbase\xspace}
\newcommand{\NPclassbase}{{\sf\bf NP}}
\newcommand{\NPclass}{\NPclassbase\xspace}

\newcommand{\complete}{\text{-complete}}

\newcommand{\NPcomplete}{\NPclassbase\complete\xspace}
\newcommand{\NPC}{\NPcomplete}

\usepackage{xspace}

\DeclareMathAlphabet{\mathscr}{OT1}{pzc}{m}{it}
\DeclareMathAlphabet{\mathfrit}{OT1}{cmfr}{m}{it}
\DeclareMathAlphabet{\mathccmit}{OT1}{ccm}{m}{it}
\DeclareMathAlphabet{\mathdj}{U}{msb}{m}{n} 
\DeclareMathAlphabet{\mathib}{OT1}{cmr}{bx}{it} 
\DeclareMathAlphabet{\mathbfit}{OT1}{cmr}{bx}{it}
\DeclareMathAlphabet{\mathbfss}{OT1}{cmss}{bx}{n}
\DeclareMathAlphabet{\mathsfm}{OT1}{cmss}{m}{n}

\usepackage{mathrsfs}

\DeclareMathOperator{\QQ}{\mathbb{Q}}
\DeclareMathOperator{\STD}{\mathrm{s}}
\DeclareMathOperator{\SHUFFLE}{\bullet}
\DeclareMathOperator{\BINTOPERM}{\mathrm{btp}}
\DeclareMathOperator{\PERMTOBIN}{\mathrm{ptb}}

\newcommand{\CrossingLR}{%
\begin{tikzpicture}[yscale=0.08,xscale=0.09,inner sep=.8pt,node distance=.25cm,>=latex']
\draw node [draw,circle,fill=black] (U1)               {};
\draw node [draw,circle,fill=black] [right of=U1] (U2) {};
\draw node [draw,circle,fill=black] [right of=U2] (U3) {};
\draw node [draw,circle,fill=black] [right of=U3] (U4) {};
\draw [->] (U1.north) .. controls ($ (U1.north) + (0,4) $) and ($ (U3.north) + (0,4) $) .. (U3.north);
\draw [<-] (U2.north) .. controls ($ (U2.north) + (0,4) $) and ($ (U4.north) + (0,4) $) .. (U4.north);
\end{tikzpicture}
}

\newcommand{\CrossingRL}{%
\begin{tikzpicture}[yscale=0.08,xscale=0.09,inner sep=.8pt,node distance=.25cm,>=latex']
\draw node [draw,circle,fill=black] (U1)               {};
\draw node [draw,circle,fill=black] [right of=U1] (U2) {};
\draw node [draw,circle,fill=black] [right of=U2] (U3) {};
\draw node [draw,circle,fill=black] [right of=U3] (U4) {};
\draw [<-] (U1.north) .. controls ($ (U1.north) + (0,4) $) and ($ (U3.north) + (0,4) $) .. (U3.north);
\draw [->] (U2.north) .. controls ($ (U2.north) + (0,4) $) and ($ (U4.north) + (0,4) $) .. (U4.north);
\end{tikzpicture}
}

\newcommand{\InclusionLL}{%
\begin{tikzpicture}[yscale=0.08,xscale=0.09,inner sep=.8pt,node distance=.25cm,>=latex']
\draw node [draw,circle,fill=black] (U1)               {};
\draw node [draw,circle,fill=black] [right of=U1] (U2) {};
\draw node [draw,circle,fill=black] [right of=U2] (U3) {};
\draw node [draw,circle,fill=black] [right of=U3] (U4) {};
+\draw [->] (U4.north) .. controls ($ (U4.north) + (0,4) $) and ($ (U1.north) + (0,4) $) .. (U1.north);
+\draw [->] (U3.north) .. controls ($ (U3.north) + (0,3) $) and ($ (U2.north) + (0,3) $) .. (U2.north);
\end{tikzpicture}
}

\newcommand{\InclusionLR}{%
\begin{tikzpicture}[yscale=0.08,xscale=0.09,inner sep=.8pt,node distance=.25cm,>=latex']
\draw node [draw,circle,fill=black] (U1)               {};
\draw node [draw,circle,fill=black] [right of=U1] (U2) {};
\draw node [draw,circle,fill=black] [right of=U2] (U3) {};
\draw node [draw,circle,fill=black] [right of=U3] (U4) {};
\draw [->] (U1.north) .. controls ($ (U1.north) + (0,4) $) and ($ (U4.north) + (0,4) $) .. (U4.north);
\draw [<-] (U2.north) .. controls ($ (U2.north) + (0,3) $) and ($ (U3.north) + (0,3) $) .. (U3.north);
\end{tikzpicture}
}

\newcommand{\InclusionRL}{%
\begin{tikzpicture}[yscale=0.08,xscale=0.09,inner sep=.8pt,node distance=.25cm,>=latex']
\draw node [draw,circle,fill=black] (U1)               {};
\draw node [draw,circle,fill=black] [right of=U1] (U2) {};
\draw node [draw,circle,fill=black] [right of=U2] (U3) {};
\draw node [draw,circle,fill=black] [right of=U3] (U4) {};
\draw [<-] (U1.north) .. controls ($ (U1.north) + (0,4) $) and ($ (U4.north) + (0,4) $) .. (U4.north);
\draw [->] (U2.north) .. controls ($ (U2.north) + (0,3) $) and ($ (U3.north) + (0,3) $) .. (U3.north);
\end{tikzpicture}
}

\newcommand{\InclusionRR}{%
\begin{tikzpicture}[yscale=0.08,xscale=0.09,inner sep=.8pt,node distance=.25cm,>=latex']
\draw node [draw,circle,fill=black] (U1)               {};
\draw node [draw,circle,fill=black] [right of=U1] (U2) {};
\draw node [draw,circle,fill=black] [right of=U2] (U3) {};
\draw node [draw,circle,fill=black] [right of=U3] (U4) {};
\draw [->] (U1.north) .. controls ($ (U1.north) + (0,4) $) and ($ (U4.north) + (0,4) $) .. (U4.north);
\draw [->] (U2.north) .. controls ($ (U2.north) + (0,3) $) and ($ (U3.north) + (0,3) $) .. (U3.north);
\end{tikzpicture}
}

\makeatletter
\newcommand{\pushright}[1]{\ifmeasuring@#1\else\omit\hfill$\displaystyle#1$\fi\ignorespaces}
\newcommand{\pushleft}[1]{\ifmeasuring@#1\else\omit$\displaystyle#1$\hfill\fi\ignorespaces}
\makeatother

\begin{document}


\title{%
Unshuffling Permutations}%

\author{%
  Samuele Giraudo \and
  St\'ephane Vialette
}
\institute{%
  Universit\'e Paris-Est, LIGM (UMR 8049), CNRS, UPEM, ESIEE Paris, ENPC,
  F-77454, Marne-la-Vallée, France\\
  \email{samuele.giraudo@univ-mlv.fr} \\
  \email{vialette@univ-mlv.fr}
}
\date{\today}

\maketitle


\begin{abstract}
    A permutation is said to be a square if it can be obtained by
    shuffling two order-isomorphic patterns. The definition is intended
    to be the natural counterpart to the ordinary shuffle of words and
    languages. In this paper, we tackle the problem of recognizing square
    permutations from both the point of view of algebra and algorithms.
    On the one hand, we present some algebraic and combinatorial
    properties of the shuffle product of permutations. We follow an
    unusual line consisting in defining the shuffle of permutations by
    means of an unshuffling operator, known as a coproduct. This
    strategy allows to obtain easy proofs for algebraic and combinatorial
    properties of our shuffle product. We besides exhibit a bijection
    between square $(213,231)$-avoiding permutations and square binary
    words. On the other hand, by using a pattern avoidance criterion on
    oriented perfect matchings, we prove that recognizing square
    permutations is \NPC.
\end{abstract}


\section{Introduction}
\label{section:Introduction}
The {\em shuffle product}, denoted $\shuffle$, is a well-known operation on words
first defined by Eilenberg and Mac Lane~\cite{Eilenberg:MacLane:1953}.
Given three words $u$, $v_1$, and $v_2$, $u$ is said to be a \emph{shuffle}
of $v_1$ and $v_2$ if it can be formed by interleaving the letters from
$v_1$ and $v_2$ in a way that maintains the left-to-right ordering of the
letters from each word.
Besides purely combinatorial questions,
the shuffle product of words naturally leads to the following computational
problems:
\begin{enumerate}
    \item \label{item:problem_1}
    Given two words $v_1$ and $v_2$, compute the set $v_1 \shuffle v_2$.
    \item \label{item:problem_2}
    Given three words $u$, $v_1$, and $v_2$, decide if $u$ is a shuffle
    of $v_1$ and $v_2$.
    \item \label{item:problem_3}
    Given words $u$, $v_1$, \dots, $v_k$, decide if $u$ is in
    $v_1 \shuffle \dots \shuffle v_k$.
    \item \label{item:problem_4}
    Given a word $u$, decide if there is a word $v$ such that $u$ is
    in $v \shuffle v$.
\end{enumerate}
Even if these problems seem similar, they radically differ in terms
of time complexity. Let us now review some facts about these. In what follows,
$n$ denotes the size of $u$ and $m_i$ denotes the size of each $v_i$.
A solution to Problem~\ref{item:problem_1} can be computed in
$O\left((m_1 + m_2) \; \binom{m_1 + m_2}{m_1}\right)$
time~\cite{Spehner:TCS:1986}. An improvement and a generalization
of~Problem~\ref{item:problem_1} has been proposed
in~\cite{Allauzen:IGM:2000}, where it is proved that given words
$v_1$, \dots, $v_k$, the iterated shuffle
$v_1 \shuffle \dots \shuffle v_k$ can be computed in
$O\left(\binom{m_1 + \dots + m_k}{m_1, \dots, m_k}\right)$
time. Problem~\ref{item:problem_2} is in \Pclass; it is indeed a
classical textbook exercise to design an efficient dynamic programming
algorithm solving it. It can be tested in $O\left(n^2 / \log(n)\right)$
time~\cite{Leeuwen:Nivat:IPL:1982}. To the best of our knowledge, the
first $O(n^2)$ time algorithm for this problem appeared
in~\cite{Mansfield:DAM:1983}. This algorithm can easily be extended to
check in polynomial-time whether or not a word is in the shuffle of any
fixed number of given words. Nevertheless, Problem~\ref{item:problem_3}
is \NPC~\cite{Mansfield:DAM:1983,Warmuth:Haussler:JCSS:1984}. This
remains true even if the ground alphabet has size
$3$~\cite{Warmuth:Haussler:JCSS:1984}. Of particular interest,
it is shown in~\cite{Warmuth:Haussler:JCSS:1984} that
Problem~\ref{item:problem_3} remains \NPC even if all the words $v_i$,
$i \in [k]$, are identical, thereby proving that, for two words $u$ and
$v$, it is \NPC to decide whether or not $u$ is in the iterated shuffle
of $v$. Again, this remains true even if the ground alphabet has size $3$.
Let us now finally focus on Problem~\ref{item:problem_4}. It is shown
in~\cite{Buss:Soltys:2014,Rizzi:Vialette:CSR:2013} that it is \NPC to
decide if a word $u$ is a \emph{square} (w.r.t. the shuffle), that is
a word $u$ with the property that there exists a word $v$ such that $u$
is a shuffle of $v$ with itself. Hence, Problem~\ref{item:problem_4}
is \NPC.
\smallskip

This paper is intended to study a natural generalization of $\shuffle$,
denoted by $\SHUFFLE$, as a shuffle of permutations. Roughly speaking,
given three permutations $\pi$, $\sigma_1$, and $\sigma_2$, $\pi$ is
said to be a {\em shuffle} of $\sigma_1$ and $\sigma_2$ if $\pi$ (viewed
as a word) is a shuffle of two words that are order-isomorphic to
$\sigma_1$ and $\sigma_2$. This shuffle product was first introduced by
Vargas~\cite{Vargas:2014} under the name of {\em supershuffle}. Our
intention in this paper is to study this shuffle product of permutations
$\SHUFFLE$ both from a combinatorial and from a computational point of
view by focusing on {\em square} permutations, that are permutations
$\pi$ being in the shuffle of a permutation $\sigma$ with itself. Many
other shuffle products on permutations appear in the literature. For
instance, in~\cite{DHT:IJAC:2002}, the authors define the {\em convolution
product} and the {\em shifted shuffle product}. For this last product,
$\pi$ is a shuffle of $\sigma_1$ and $\sigma_2$ if $\pi$ is in the shuffle,
as words, of $\sigma_1$ and the word obtained by incrementing all the
letters of $\sigma_2$ by the size of $\sigma_1$. It is a simple exercise
to prove that, given three permutations $\pi$, $\sigma_1$, and $\sigma_2$,
deciding if $\pi$ is in the shifted shuffle of $\sigma_1$ and $\sigma_2$
is in~\Pclass.
\smallskip

This paper is organized as follows. In
Section~\ref{section:Shuffle product on permutations} we provide a
precise definition of $\SHUFFLE$. This definition passes through the
preliminary definition of an operator~$\Delta$, allowing to
{\em unshuffle} permutations. This operator is in fact a coproduct,
endowing the linear span of all permutations with a coalgebra structure
(see~\cite{Joni:Rota:1979} or~\cite{Grinberg:Reiner:2014} for the
definition of these algebraic structures). By duality, the unshuffling
operator $\Delta$ leads to the definition of our shuffle operation on
permutations. This approach has many advantages. First, some
combinatorial properties of $\SHUFFLE$ depend on properties of $\Delta$
and are more easy to prove on the coproduct side. Second, this way of
doing allows to obtain a clear description of the multiplicities of the
elements appearing in the shuffle of two permutations, which are worthy
of interest from a combinatorial point of view.
Section~\ref{section:Binary square words and permutations} is devoted to
showing that the problems related to the shuffle of words has links with
the shuffle of permutations. In particular, we show that binary words
that are square are in one-to-one correspondence with square
permutations avoiding some patterns
(Proposition~\ref{prop:bijection_binary_to_permutations_squares}). Next,
Section~\ref{section:Algebraic issues} presents some algebraic and
combinatorial properties of~$\SHUFFLE$. We show that $\SHUFFLE$ is
associative and commutative
(Proposition~\ref{prop:shuffle_associative_commutative}), and that if a
permutation is a square, its mirror, complement, and inverse are also
squares (Proposition~\ref{prop:square_stability}). Finally,
Section~\ref{section:Algorithmic issues} presents the most important result
of this paper: the fact that deciding if a permutation is a square is
\NPC (Proposition~\ref{proposition:hardness}). This result is obtained
by exhibiting a reduction from the pattern involvement
problem~\cite{Bose:Buss:Lubiw:1998} which is \NPC.


\section{Notations}
\label{section:Notations}

If $S$ is a finite set, the cardinality of $S$ is denoted by $|S|$,
and if $P$ and $Q$ are two disjoint sets, $P \sqcup Q$ denotes the
disjoint union of $P$ and $Q$. For any nonnegative integer $n$, $[n]$
is the set $\{1, \dots, n\}$.

We follow the usual terminology on words~\cite{ChoffrutKarhumaki1997}.
Let us recall here the most important ones. Let $u$ be a word. The
length of $u$ is denoted by $|u|$. The {\em empty word}, the only word
of null length, is denoted by $\epsilon$. We denote by $\widetilde{u}$
the {\em mirror image} of $u$, that is the word
$u_{|u|} u_{|u| - 1} \dots u_1$. If $P$ is a subset of $[|u|]$, $u_{|P}$
is the subword of $u$ consisting in the letters of $u$ at the positions
specified by the elements of $P$. If $u$ is a word of integers and $k$
is an integer, we denote by $u[k]$ the word obtained by incrementing by
$k$ all letters of $u$. The {\em shuffle} of two words $u$ and $v$ is
the set recursively defined by
$u \shuffle \epsilon = \{u\} = \epsilon \shuffle u$ and
$ua \shuffle vb = (u \shuffle vb)a \cup (ua \shuffle v)b$, were $a$ and
$b$ are letters. A word $u$ is a {\em square} if there exists a word $v$
such that $u$ belongs to $v \shuffle v$.

We denote by $S_n$ the set of permutations of size $n$ and by $S$ the
set of all permutations. In this paper, permutations of a size $n$ are
specified by words of length $n$ on the alphabet $[n]$ and without
multiple occurrence of a letter, so that all above definitions about
words remain valid on permutations. The only difference lies on the
fact that we shall denote by $\pi(i)$ (instead of $\pi_i$) the $i$-th
letter of any permutation $\pi$. For any nonnegative integer $n$, we
write $\nearrow_{n}$ (resp. $\searrow_{n}$) for the permutation
$1 2 \dots n$ (resp. $n\,(n-1) \dots 1$). If $\pi$ is a permutation of
$S_n$, we denote by $\bar \pi$ the {\em complement} of $\pi$, that is
the permutation satisfying $\bar \pi(i) = n - \pi(i) + 1$ for all
$i \in [n]$. The {\em inverse} of $\pi$ is denoted by $\pi^{-1}$.

If $u$ is a word of integers without multiple occurrences of a same
letter, $\STD(u)$ is the {\em standardized} of $u$, that is the unique
permutation of the same size as $u$ such that for all $i, j \in [|u|]$,
$u_i < u_j$ if and only if $\STD(u)(i) < \STD(u)(j)$. In particular, the
image of the map $\STD$ is the set $S$ of all permutations. Two words $u$
and $v$ having the same standardized are {\em order-isomorphic}. If
$\sigma$ is a permutation, there is an {\em occurrence} of (the
{\em pattern}) $\sigma$ in $\pi$ if there is a set $P$ of indexes of
letters of $\pi$ such that $\sigma$ and $\pi_{|P}$ are order-isomorphic.
When $\pi$ does not admit any occurrence of $\sigma$, $\pi$ {\em avoids}
$\sigma$. The set of permutations of size $n$ avoiding $\sigma$ is
denoted by $S_n(\sigma)$.

Let us now provide some definitions about graphs and oriented perfect
matchings that are used in the sequel. If $G$ is an oriented graph
without loops, two different edges of $G$ are {\em independent} if they
do not share any common vertex. We say that $G$ is an
{\em oriented matching} if all edges of $G$ are pairwise independent.
Moreover, $G$ is {\em perfect} if any vertex of $G$ belongs to at least
one arc. For any permutation $\pi$ of $S_n$, an
{\em oriented perfect matching on $\pi$} is an oriented perfect matching
$\mathcal{M}$ on the set of vertices $[n]$. In the sequel, we shall
consider a natural notion of pattern avoidance in oriented perfect
matchings on permutations. For instance, an oriented perfect matching
$\mathcal{M}$ on a permutation $\pi$ {\em admits an occurrence} of the
pattern $\CrossingRL$ if there are four positions $i < j < k < \ell$ in
$\pi$ such that $(\pi(k), \pi(i))$ and $(\pi(j), \pi(\ell))$ are arcs of
$\mathcal{M}$. When $\mathcal{M}$ does not admit any occurrence of a
pattern $\mathcal{P}$, we say that $\mathcal{M}$ {\em avoids}~$\mathcal{P}$.
The definition naturally extends to sets of patterns:
$\mathcal{M}$ {\em avoids}~$P=\{\mathcal{P}_i : 1 \leq i \leq k\}$
if it avoids every pattern~$\mathcal{P}_i$.



\section{Shuffle product on permutations}
\label{section:Shuffle product on permutations}

The purpose of this section is to define a shuffle product $\SHUFFLE$
on permutations. Recall that a first definition of this product was
provided by Vargas~\cite{Vargas:2014}. To present an alternative
definition of this product adapted to our study, we shall first define
a coproduct denoted by $\Delta$, enabling to unshuffle permutations.
By duality, $\Delta$ implies the definition of $\SHUFFLE$. The reason
why we need to pass by the definition of $\Delta$ to define $\SHUFFLE$
is justified by the fact that a lot of properties of $\SHUFFLE$ depend
of properties of $\Delta$, and that this strategy allows to write concise
and clear proofs of them. We invite the reader unfamiliar with the
concepts of coproduct and duality to consult~\cite{Joni:Rota:1979}
or~\cite{Grinberg:Reiner:2014}.

Let us denote by $\QQ[S]$ the linear span of all permutations. We
define a linear coproduct $\Delta$ on $\QQ[S]$ in the following way. For
any permutation $\pi$, we set
\begin{equation} \label{equ:unshuffling_coproduct}
    \Delta(\pi) =
    \sum_{P_1 \sqcup P_2 = [|\pi|]}
    \STD\left(\pi_{|P_1}\right) \otimes \STD\left(\pi_{|P_2}\right).
\end{equation}
We call $\Delta$ the {\em unshuffling coproduct of permutations}. For
instance,
\begin{equation}
    \Delta(213) =
    \epsilon \otimes 213 + 2 \cdot 1 \otimes 12 +
    1 \otimes 21 + 2 \cdot 12 \otimes 1 + 21 \otimes 1
    + 213 \otimes \epsilon,
\end{equation}
\begin{equation}
    \Delta(1234) =
    \epsilon \otimes 1234 + 4 \cdot 1 \otimes 123 + 6 \cdot 12 \otimes 12 +
    4 \cdot 123 \otimes 1 + 1234 \otimes \epsilon,
\end{equation}
\begin{equation}\begin{split} \label{equ:example_unshuffling_coproduct}
    \Delta(1432) & =
    \epsilon \otimes 1432 + 3 \cdot {\bf 1 \otimes 132} + 1 \otimes 321 +
    3 \cdot 12 \otimes 21 \\ & + 3 \cdot 21 \otimes 12
    + 3 \cdot 132 \otimes 1 + 321 \otimes 1 + 1432 \otimes \epsilon.
\end{split}\end{equation}
Observe that the coefficient of the tensor $1 \otimes 132$ is $3$
in~\eqref{equ:example_unshuffling_coproduct} because there are exactly
three ways to extract from the permutation $1432$ two disjoint subwords
respectively order-isomorphic to the permutations $1$ and $132$.

As announced, let us now use $\Delta$ to define a shuffle product on
permutations. As any coproduct, $\Delta$ leads to the definition of a
product obtained by duality in the following way.
From~\eqref{equ:unshuffling_coproduct}, for any permutation $\pi$, we
have
\begin{equation} \label{equ:schematic_coproduct}
    \Delta(\pi) =
    \sum_{\sigma, \nu \in S} \lambda_{\sigma, \nu}^\pi \;
    \sigma \otimes \nu,
\end{equation}
where the $\lambda_{\sigma, \nu}^\pi$ are nonnegative integers. Now,
by definition of duality, the dual product of $\Delta$, denoted by
$\SHUFFLE$, is a linear binary product on $\QQ[S]$. It satisfies, for
any permutations $\sigma$ and $\nu$,
\begin{equation}
    \sigma \SHUFFLE \nu =
    \sum_{\pi \in S}
    \lambda_{\sigma, \nu}^\pi \; \pi,
\end{equation}
where the coefficients $\lambda_{\sigma, \nu}^\pi$ are the ones
of~\eqref{equ:schematic_coproduct}. We call $\SHUFFLE$ the
{\em shuffle product of permutations}. For instance,
\begin{equation}\begin{split} \label{equ:example_shuffle_product}
    12 \SHUFFLE 21 & =
    1243 + 1324 + 2 \cdot 1342 + 2 \cdot 1423 + 3 \cdot {\bf 1432} +
    2134 + 2 \cdot 2314 \\
    & + 3 \cdot 2341 + 2413 + 2 \cdot 2431 + 2 \cdot 3124 + 3142 +
    3 \cdot 3214 + 2 \cdot 3241 \\
    & + 3421 + 3 \cdot 4123 + 2 \cdot 4132 + 2 \cdot 4213 + 4231 + 4312.
\end{split}\end{equation}
Observe that the coefficient $3$ of the permutation $1432$
in~\eqref{equ:example_shuffle_product} comes from the fact that the
coefficient of the tensor $12 \otimes 21$ is $3$
in~\eqref{equ:example_unshuffling_coproduct}.

Intuitively, this product shuffles the values and the positions of the
letters of the permutations. One can observe that the empty permutation
$\epsilon$ is a unit for $\SHUFFLE$ and that this product is graded by
the sizes of the permutations ({\em i.e.}, the product of a permutation
of size $n$ with a permutation of size $m$ produces a sum of permutations
of size $n + m$).

We say that a permutation $\pi$ {\em appears} in the shuffle
$\sigma \SHUFFLE \nu$ of two permutations $\sigma$ and $\nu$ if the
coefficient $\lambda_{\sigma, \nu}^\pi$ defined above is different from
zero. In a more combinatorial way, this is equivalent to say that there
are two sets $P_1$ and $P_2$ of disjoints indexes of letters of $\pi$
satisfying $P_1 \sqcup P_2 = [|\pi|]$ such that the subword $\pi_{|P_1}$
is order-isomorphic to $\sigma$ and the subword $\pi_{|P_2}$ is
order-isomorphic to $\nu$.

A permutation $\pi$ is a {\em square} if there is a permutation
$\sigma$ such that $\pi$ appears in $\sigma \SHUFFLE \sigma$.
In this case, we say that $\sigma$ is a {\em square root} of $\pi$.
Equivalently, $\pi$ is a square with $\sigma$ as square root if and only
if in the expansion of $\Delta(\pi)$, there is a tensor
$\sigma \otimes \sigma$ with a nonzero coefficient. In a more
combinatorial way, this is equivalent to saying that there are two sets
$P_1$ and $P_2$ of disjoints indexes of letters of $\pi$ satisfying
$P_1 \sqcup P_2 = [|\pi|]$ such that the subwords $\pi_{|P_1}$ and
$\pi_{|P_2}$ are order-isomorphic. Computer experiments give us the
first numbers of square permutations with respects to their size, which
are, from size $0$ to $10$,
\begin{equation}
    1, 0, 2, 0, 20, 0, 504, 0, 21032, 0, 1293418.
\end{equation}
This sequence (and its subsequence obtained by removing the $0$'s) is for
the time being not listed in~\cite{Slo}. The square permutations of
sizes $0$ to $4$ are
\smallskip

\begin{tabular}{c|c|c}
    Size $0$ \, & Size $2$ & Size $4$ \\ \hline
    \multirow{2}{*}{\, $\epsilon$} &
    \multirow{2}{*}{\, $12$, $21$ \,} &
    $1234$, $1243$, $1423$, $1324$, $1342$, $4132$, $3124$, $3142$,
    $3412$, $4312$, \\
    & & $2134$, $2143$, $2413$, $4213$, $2314$, $2431$, $4231$, $3241$,
    $3421$, $4321$
\end{tabular}


\section{Binary square words and permutations}
\label{section:Binary square words and permutations}
In this section, we shall show that the square binary words are in
one-to-one correspondence with square permutations avoiding some
patterns. This property establishes a link between the shuffle of binary
words and our shuffle of permutations and allows to obtain a new
description of square binary words.

Let $u$ be a binary word of length $n$ with $k$ occurrences of $0$.
We denote by $\BINTOPERM$ (Binary word To Permutation) the map sending
any such word $u$ to the permutation obtained by replacing from left to
right each occurrence of $0$ in $u$ by $1$, $2$, \dots, $k$, and from
right to left each occurrence of $1$ in $u$ by $k + 1$, $k + 2$, \dots, $n$.
For instance,
\begin{equation}
    \BINTOPERM({\bf 1}00{\bf 1}0{\bf 1}{\bf 1}0{\bf 1}000) =
    {\bf C} 12 {\bf B} 3 {\bf A} {\bf 9} 4 {\bf 8} 567,
\end{equation}
where $\mathrm{A}$, $\mathrm{B}$, and $\mathrm{C}$ respectively stand
for $10$, $11$, and $12$. Observe that for any nonempty permutation
$\pi$ in the image of $\BINTOPERM$, there is exactly one binary word $u$
such that $\BINTOPERM(u0) = \BINTOPERM(u1) = \pi$. In support of this
observation, when $\pi$ has an even size, we denote by $\PERMTOBIN(\pi)$
(Permutation To Binary word) the word $ua$ such that $|ua|_0$ and $|ua|_1$
are both even, where $a \in \{0, 1\}$.

\begin{proposition} \label{prop:bijection_binary_to_permutations_squares}
    For any $n \geq 0$, the map $\BINTOPERM$ restricted to the set of
    square binary words of length $2n$ is a bijection between this last
    set and the set of square permutations of size $2n$ avoiding the
    patterns $213$ and $231$.
\end{proposition}
\begin{proof}
    [of Proposition~\ref{prop:bijection_binary_to_permutations_squares}]
    The statement of the proposition is a consequence of the
    following claims implying that $\PERMTOBIN$ is the inverse
    map of $\BINTOPERM$ over the set of square binary words.
    \begin{claim} \label{claim:binary_to_permutation_avoiding}
        The image of $\BINTOPERM$ is the set of all permutations
        avoiding $213$ and $231$.
    \end{claim}
    \begin{proof}[of Claim~\ref{claim:binary_to_permutation_avoiding}]
        Let us first show that the image of $\BINTOPERM$ contains only
        permutations avoiding $213$ and $231$. Let $u$ be a binary word,
        $\pi = \BINTOPERM(u)$, and $P_0$ (resp. $P_1$) be the set of the
        positions of the occurrences of $0$ (resp. $1$) in $u$. By
        definition of $\BINTOPERM$, from left to right, the subword
        $v = \pi_{|P_0}$ is increasing and the subword $w = \pi_{|P_1}$
        is decreasing, and all letters of $w$ are greater than those
        of~$v$. Now, assume that $\pi$ admits an occurrence of $213$.
        Then, since $v$ is increasing and $w$ is decreasing, there is an
        occurrence of $3$ (resp. $13$, $23$) in $v$ and a relative
        occurrence of $21$ (resp. $2$, $1$). All these three cases
        contradict the fact that all letters of $w$ are greater than
        those of $v$. A similar argument shows that $\pi$ avoids~$231$
        as well.
        \smallskip

        Finally, observe that any permutation $\pi$ avoiding $213$ and
        $231$ necessarily starts by the smallest possible letter or the
        greatest possible letter. This property is then true for the
        suffix of $\pi$ obtained by deleting its first letter,
        and so on for all of its suffixes. Thus, by
        replacing each letter $a$ of $\pi$ by $0$ (resp. $1$) if $a$ has the
        role of a smallest (resp. greatest) letter, one obtains a binary
        word $u$ such that $\BINTOPERM(u) = \pi$. Hence, all permutations
        avoiding $213$ and $231$ are in the image of $\BINTOPERM$.
        \qed
    \end{proof}

    \begin{claim} \label{claim:square_binary_to_square_permutation}
        If $u$ is a square binary word, $\BINTOPERM(u)$ is a square
        permutation.
    \end{claim}
    \begin{proof}[of Claim~\ref{claim:square_binary_to_square_permutation}]
        Since $u$ is a square binary word, there is a binary word $v$
        such that $u \in v \shuffle v$. Then, there are two disjoint
        sets $P$ and $Q$ of positions of letters of $u$ such that
        $u_{|P} = v = u_{|Q}$. Now, by definition of $\BINTOPERM$, the
        words $\BINTOPERM(u)_{|P}$ and $\BINTOPERM(u)_{|Q}$ have the
        same standardized $\sigma$. Hence, and by definition of
        the shuffle product of permutations, $\BINTOPERM(u)$ appears in
        $\sigma \SHUFFLE \sigma$, showing that $\BINTOPERM(u)$ is a
        square permutation.
        \qed
    \end{proof}

    \begin{claim} \label{claim:square_permutation_to_square_binary}
        If $\pi$ is a square permutation avoiding $213$ and $231$,
        $\PERMTOBIN(\pi)$ is a square binary word.
    \end{claim}
    \begin{proof}[of Claim~\ref{claim:square_permutation_to_square_binary}]
        Let $\pi$ be a square permutation avoiding $213$ and $231$. By
        Claim~\ref{claim:binary_to_permutation_avoiding}, $\pi$ is in
        the image of $\BINTOPERM$ and hence, $u = \PERMTOBIN(\pi)$ is a
        well-defined binary word. Since $\pi$ is a square permutation,
        there are two disjoint sets $P_1$ and $P_2$ of indexes of letters
        of $\pi$ such that $\pi_{|P_1}$ and $\pi_{|P_2}$ are
        order-isomorphic. This implies, by the definitions of $\BINTOPERM$
        and $\PERMTOBIN$, that $u_{|P_1} = u_{|P_2}$, showing that $u$
        is a square binary word.
        \qed
    \end{proof}
    \qed
\end{proof}

The number of square binary words is Sequence A191755 of~\cite{Slo}
beginning by
\begin{equation}
1, 0, 2, 0, 6, 0, 22, 0, 82, 0, 320, 0, 1268, 0, 5102, 0, 020632.
\end{equation}
According to
Proposition~\ref{prop:bijection_binary_to_permutations_squares}, this
is also the sequence enumerating square permutations avoiding $213$
and $231$.


\section{Algebraic issues}
\label{section:Algebraic issues}
The aim of this section is to establish some of properties of the
shuffle product of permutations $\SHUFFLE$. It is worth to note that, as
we will see, algebraic properties of the unshuffling coproduct $\Delta$
of permutations defined in
Section~\ref{section:Shuffle product on permutations} lead to
combinatorial properties of $\SHUFFLE$.

\begin{proposition} \label{prop:shuffle_associative_commutative}
    The shuffle product $\SHUFFLE$ of permutations is associative and
    commutative.
\end{proposition}
\begin{proof}[of Proposition~\ref{prop:shuffle_associative_commutative}]
    To prove the associativity of $\SHUFFLE$, it is convenient to show
    that its dual coproduct $\Delta$ is coassociative, that is
    \begin{equation}
        (\Delta \otimes I) \Delta = (I \otimes \Delta) \Delta,
    \end{equation}
    where $I$ denotes the identity map. This strategy relies on the fact
    that a product is associative if and only if its dual coproduct is
    coassociative. For any permutation $\pi$, we have
    \begin{equation} \begin{split}
    \label{equ:shuffle_associative_commutative}
        (\Delta \otimes I) \Delta(\pi) & =
        (\Delta \otimes I)
        \sum_{P_1 \sqcup P_2 = [|\pi|]}
        \STD\left(\pi_{|P_1}\right) \otimes \STD\left(\pi_{|P_2}\right) \\
        & =
        \sum_{P_1 \sqcup P_2 = [|\pi|]}
        \Delta\left(\STD\left(\pi_{|P_1}\right)\right)
        \otimes I\left(\STD\left(\pi_{|P_2}\right)\right) \\
        & =
        \sum_{P_1 \sqcup P_2 = [|\pi|]} \;
        \sum_{Q_1 \sqcup Q_2 = [|P_1|]}
        \STD\left(\STD\left(\pi_{|P_1}\right)_{|Q_1}\right)
        \otimes
        \STD\left(\STD\left(\pi_{|P_1}\right)_{|Q_2}\right)
        \otimes \STD\left(\pi_{|P_2}\right) \\
        & =
        \sum_{P_1 \sqcup P_2 \sqcup P_3 = [|\pi|]}
        \STD\left(\pi_{|P_1}\right) \otimes
        \STD\left(\pi_{|P_2}\right) \otimes
        \STD\left(\pi_{|P_3}\right).
    \end{split} \end{equation}
    An analogous computation shows that $(I \otimes \Delta) \Delta(\pi)$
    is equal to the last member
    of~\eqref{equ:shuffle_associative_commutative}, whence the
    associativity of $\SHUFFLE$.
    \smallskip

    Finally, to prove the commutativity of $\SHUFFLE$, we shall show
    that $\Delta$ is cocommutative, that is for any permutation $\pi$,
    if in the expansion of $\Delta(\pi)$ there is a tensor
    $\sigma \otimes \nu$ with a coefficient $\lambda$, there is in the
    same expansion the tensor $\nu \otimes \sigma$ with the same
    coefficient $\lambda$. Clearly, a product is commutative if and only
    if its dual coproduct is cocommutative. Now, from the
    definition~\eqref{equ:unshuffling_coproduct} of $\Delta$, one
    observes that if the pair $(P_1, P_2)$ of subsets of $[|\pi|]$
    contributes to the coefficient of
    $\STD\left(\pi_{|P_1}\right) \otimes \STD\left(\pi_{|P_2}\right)$,
    the pair $(P_2, P_1)$ contributes to the coefficient of
    $\STD\left(\pi_{|P_2}\right) \otimes \STD\left(\pi_{|P_1}\right)$.
    This shows that $\Delta$ is cocommutative and hence, that $\SHUFFLE$
    is commutative.
    \qed
\end{proof}

Proposition~\ref{prop:shuffle_associative_commutative} implies that
$\QQ[S]$ endowed with the unshuffling coproduct $\Delta$ is a
coassociative cocommutative coalgebra, or in an equivalent way, that
$\QQ[S]$ endowed with the shuffle product $\SHUFFLE$ is an associative
commutative algebra.

\begin{lemma} \label{lem:endomorphisms}
    The three linear maps
    \begin{equation}
        \phi_1, \phi_2, \phi_3 : \QQ[S] \to \QQ[S]
    \end{equation}
    linearly sending a permutation $\pi$ to, respectively,
    $\widetilde{\pi}$, $\bar \pi$, and $\pi^{-1}$ are endomorphims of
    associative algebras.
\end{lemma}

We now use the algebraic properties of $\SHUFFLE$ exhibited by
Lemma~\ref{lem:endomorphisms} to obtain combinatorial properties
of square permutations.

\begin{proposition} \label{prop:square_stability}
    Let $\pi$ be a square permutation and $\sigma$ be a square root of
    $\pi$. Then,
    \begin{enumerate}[label={\it (\roman*)},fullwidth]
        \item \label{item:square_stability_1}
        the permutation $\widetilde{\pi}$ is a square and
        $\widetilde{\sigma}$ is one of its square roots;
        \item \label{item:square_stability_2}
        the permutation $\bar \pi$ is a square and $\bar \sigma$ is one of
        its square roots;
        \item \label{item:square_stability_3}
        the permutation $\pi^{-1}$ is a square and $\sigma^{-1}$ is one of
        its square roots.
    \end{enumerate}
\end{proposition}
\begin{proof}[of Proposition~\ref{prop:square_stability}]
    All statements~\ref{item:square_stability_1},
    \ref{item:square_stability_2}, and~\ref{item:square_stability_3} are
    consequences of Lemma~\ref{lem:endomorphisms}. Indeed,
    since $\pi$ is a square permutation and $\sigma$ is a square root of
    $\pi$, by definition, $\pi$ appears in the product
    $\sigma \SHUFFLE \sigma$. Now, by Lemma~\ref{lem:endomorphisms},
    for any $j = 1, 2, 3$, since $\phi_j$ is a morphism of associative
    algebras from $\QQ[S]$ to $\QQ[S]$, $\phi_j$ commutes with the
    shuffle product of permutations $\SHUFFLE$. Hence, in particular,
    one has
    \begin{equation}
        \phi_j(\sigma \SHUFFLE \sigma) =
        \phi_j(\sigma) \SHUFFLE \phi_j(\sigma).
    \end{equation}
    Then, since $\pi$ appears in $\sigma \SHUFFLE \sigma$, $\phi_j(\pi)$
    appears in $\phi_j(\sigma \SHUFFLE \sigma)$ and appears also in
    $\phi_j(\sigma) \SHUFFLE \phi_j(\sigma)$. This shows that
    $\phi_j(\sigma)$ is a square root of $\phi_j(\pi)$ and
    implies~\ref{item:square_stability_1}, \ref{item:square_stability_2},
    and~\ref{item:square_stability_3}.
    \qed
\end{proof}

Let us make an observation about Wilf-equivalence classes of permutations
restrained on square permutations. Recall that two permutations $\sigma$
and $\nu$ of the same size are {\em Wilf equivalent} if
$\# S_n(\sigma) = \# S_n(\nu)$ for all $n \geq 0$. The
well-known~\cite{Simion:Schmidt:EJC:1985} fact that there is a single
Wilf-equivalence class of permutations of size $3$ together with
Proposition~\ref{prop:square_stability} imply that $123$ and $321$ are
in the same Wilf-equivalence class of square permutations, and that
$132$, $213$, $231$, and $312$ are in the same Wilf-equivalence class of
square permutations. Computer experiments show us that there are two
Wilf-equivalence classes of square permutations of size $3$. Indeed, the
number of square permutations avoiding $123$ begins by
\begin{equation} \label{equ:sequence_square_123}
    1, 0, 2, 0, 12, 0, 118, 0, 1218, 0, 14272,
\end{equation}
while the number of square permutations avoiding $132$ begins by
\begin{equation} \label{equ:sequence_square_132}
    1, 0, 2, 0, 11, 0, 84, 0, 743, 0, 7108.
\end{equation}

Besides, an other consequence of Proposition~\ref{prop:square_stability}
is that its makes sense to enumerate the sets of square permutations
quotiented by the operations of mirror image, complement, and
inverse. The sequence enumerating these sets begins by
\begin{equation} \label{equ:sequence_square_classes}
    1, 0, 1, 0, 6, 0, 81, 0, 2774, 0, 162945.
\end{equation}

All Sequences~\eqref{equ:sequence_square_123}, \eqref{equ:sequence_square_132},
and~\eqref{equ:sequence_square_classes} (and their subsequences obtained
by removing the $0$s) are for the time being not listed in~\cite{Slo}.


\section{Algorithmic issues}
\label{section:Algorithmic issues}

This section is devoted to proving hardness of recognizing square
permutations. In the same way as happens with words, we shall use a
linear graph framework where deciding whether a permutation is a square
reduces to computing some specific matching in the associated linear
graph~\cite{Buss:Soltys:2014,Rizzi:Vialette:CSR:2013}. We have, however,
to deal with oriented perfect matchings. The needed properties read
as follows (see Fig.~\ref{fig:example containment-free matching}).
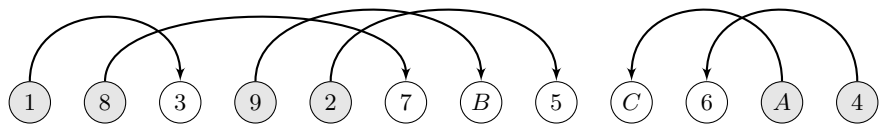
\begin{figure}[ht!]
    \centering
    \begin{tikzpicture}[scale=0.5,inner sep=2pt,node distance=1cm]
        \draw node [draw,circle,minimum size=.55cm,fill=black!10](U01){$1$};
        \draw node [draw,circle,minimum size=.55cm,fill=black!10,right of=U01]
            (U02) {$8$};
        \draw node [draw,circle,minimum size=.55cm,right of=U02](U03){$3$};
        \draw node [draw,circle,minimum size=.55cm,fill=black!10,right of=U03]
            (U04) {$9$};
        \draw node [draw,circle,minimum size=.55cm,fill=black!10,right of=U04]
            (U05) {$2$};
        \draw node [draw,circle,minimum size=.55cm,right of=U05](U06){$7$};
        \draw node [draw,circle,minimum size=.55cm,right of=U06](U07){$B$};
        \draw node [draw,circle,minimum size=.55cm,right of=U07](U08){$5$};
        \draw node [draw,circle,minimum size=.55cm,right of=U08](U09){$C$};
        \draw node [draw,circle,minimum size=.55cm,right of=U09](U10){$6$};
        \draw node [draw,circle,minimum size=.55cm,fill=black!10,right of=U10]
            (U11) {$A$};
        \draw node [draw,circle,minimum size=.55cm,fill=black!10,right of=U11]
            (U12) {$4$};
        \draw [thick,->,>=latex']
        (U01.north) .. controls ($ (U01.north) + (0,2.25) $)
            and ($ (U03.north) + (0,2.25) $) .. (U03.north);
        \draw [thick,->,>=latex']
        (U02.north) .. controls ($ (U02.north) + (0,2.25) $)
            and ($ (U06.north) + (0,2.25) $) .. (U06.north);
        \draw [thick,->,>=latex']
        (U04.north) .. controls ($ (U04.north) + (0,2.5) $)
            and ($ (U07.north) + (0,2.5) $) .. (U07.north);
        \draw [thick,->,>=latex']
        (U05.north) .. controls ($ (U05.north) + (0,2.5) $)
            and ($ (U08.north) + (0,2.5) $) .. (U08.north);
        \draw [thick,->,>=latex']
        (U11.north) .. controls ($ (U11.north) + (0,2.5) $)
            and ($ (U09.north) + (0,2.5) $) .. (U09.north);
        \draw [thick,->,>=latex']
        (U12.north) .. controls ($ (U12.north) + (0,2.5) $)
            and ($ (U10.north) + (0,2.5) $) .. (U10.north);
    \end{tikzpicture}
    \caption{\label{fig:example containment-free matching}%
        An oriented perfect matching $\mathcal{M}$ on the permutation
        $\pi = 183927B5C6A4$ satisfying the properties $\mathbf{P_1}$
        and $\mathbf{P_2}$. From $\mathcal{M}$, it follows that $\pi$ is
        a square as it appears in the shuffle of $1892A4$ and $37B5C6$,
        both being order-isomorphic to $145263$.
    }
\end{figure}

\begin{definition}[Property $\mathbf{P_1}$]
  \label{definition:Property P_1}
  Let $\pi$ be a permutation. An oriented perfect matching $\mathcal{M}$
  on $\pi$ is said to have property $\mathbf{P_1}$ if it avoids all the
  six patterns \InclusionLL, \InclusionLR, \InclusionRL,
  \InclusionRR, \CrossingLR, and \CrossingRL.
\end{definition}

\begin{definition}[Property $\mathbf{P_2}$]
  \label{definition:Property P_2}
  Let $\pi$ be a permutation. An oriented perfect matching $\mathcal{M}$
  on $\pi$ is said to have property $\mathbf{P_2}$ if, for any two
  distinct arcs $(\pi(a), \pi(a'))$ and $(\pi(b), \pi(b'))$ in $\mathcal{M}$,
  we have $\pi(a) < \pi(b)$ if and only if $\pi(a') < \pi(b')$.
\end{definition}

The rationale for introducing properties $\mathbf{P_1}$ and $\mathbf{P_2}$
stems from the following lemma.

\begin{lemma}
  \label{lemma:matching}
  Let $\pi$ be a permutation. The following statements are equivalent:
  \begin{enumerate}
    \item The permutation $\pi$ is a square.
    \item There exists an oriented perfect matching $\mathcal{M}$
    on $\pi$ satisfying~$\mathbf{P_1}$ and~$\mathbf{P_2}$.
  \end{enumerate}
\end{lemma}

Let $\pi$ be a permutation. For the sake of clarity, we will say that a
bunch of consecutive positions $P$ of $\pi$ is \emph{above} (resp.
\emph{below}) above another bunch of consecutive positions $P'$ in $\pi$
if $\pi(i) > \pi(j)$ (resp. $\pi(i) < \pi(j)$) for every $i \in P$ and
every $j \in P'$. For example, $\sigma_1$ is above $\sigma_2$ (in an
equivalent manner, $\sigma_2$ is below $\sigma_1$) in
Fig.~\ref{subfig:increasing before and above decreasing}, whereas
$\sigma_1$ is below $\sigma_2$ (in an equivalent manner, $\sigma_2$ is
above $\sigma_1$) in Fig.~\ref{subfig:decreasing before and below
increasing}.
\begin{figure}[t!]
  \centering
  \subfigure[%
    An increasing pattern before and above a decreasing pattern.
  ]{%
    \begin{tikzpicture}
      [
        scale=0.2,
        label/.style={anchor=base}
      ]
      \draw[step=1cm,black!30,ultra thin,fill=black!10] (-0.2,5.8) grid (5.2,11.2);
      \foreach \x/\y in {0/6,1/7,2/8,3/9,4/10,5/11} {
          \draw [fill=black] (\x,\y) circle (0.1);
      }
      \draw[step=1cm,black!30,ultra thin,fill=black!10] (5.8,-0.2) grid (11.2,5.2);
      \foreach \x/\y in {6/5,7/4,8/3,9/2,10/1,11/0} {
          \draw [fill=black] (\x,\y) circle (0.1);
      }
      \node [label] (a) at (1,4) {$a$};
      \node [label] (b) at (4,4) {$b$};
      \node [label] (ap) at (7,-2) {$a'$};
      \node [label] (bp) at (10,-2) {$b'$};
      \node (nu1) at (2.5,3) {$\sigma_1$};
      \node (nu2) at (8.5,-3) {$\sigma_2$};
      \draw [line width=1pt,->,>=latex']
      (1,7) .. controls +(0,12) and +(0,10) .. (7,4);
      \draw [line width=1pt,->,>=latex']
      (4,10) .. controls +(0,12) and +(0,10) .. (10,1);
    \end{tikzpicture}
    \;
    \begin{tikzpicture}
      [
        scale=0.2,
        label/.style={anchor=base}
      ]
      \draw[step=1cm,black!30,ultra thin,fill=black!10] (-0.2,5.8) grid (5.2,11.2);
      \foreach \x/\y in {0/6,1/7,2/8,3/9,4/10,5/11} {
          \draw [fill=black] (\x,\y) circle (0.1);
      }
      \draw[step=1cm,black!30,ultra thin,fill=black!10] (5.8,-0.2) grid (11.2,5.2);
      \foreach \x/\y in {6/5,7/4,8/3,9/2,10/1,11/0} {
          \draw [fill=black] (\x,\y) circle (0.1);
      }
      \node [label] (a) at (1,4) {$a'$};
      \node [label] (b) at (4,4) {$b'$};
      \node [label] (ap) at (7,-2) {$a$};
      \node [label] (bp) at (10,-2) {$b$};
      \node (nu1) at (2.5,3) {$\sigma_1$};
      \node (nu2) at (8.5,-3) {$\sigma_2$};
      \draw [line width=1pt,<-,>=latex']
      (1,7) .. controls +(0,12) and +(0,10) .. (7,4);
      \draw [line width=1pt,<-,>=latex']
      (4,10) .. controls +(0,12) and +(0,10) .. (10,1);
    \end{tikzpicture}
    \label{subfig:increasing before and above decreasing}
  }
  \qquad
  \subfigure[%
    A decreasing pattern before and below an increasing pattern.
  ]{%
  \begin{tikzpicture}
    [
      scale=0.2,
      label/.style={anchor=base}
    ]
    \draw[step=1cm,black!30,ultra thin,fill=black!10] (-0.2,-0.2) grid (5.2,5.2);
    \foreach \x/\y in {0/5,1/4,2/3,3/2,4/1,5/0} {
        \draw [fill=black] (\x,\y) circle (0.1);
    }
    \draw[step=1cm,black!30,ultra thin,fill=black!10] (5.8,5.8) grid (11.2,11.2);
    \foreach \x/\y in {6/6,7/7,8/8,9/9,10/10,11/11} {
        \draw [fill=black] (\x,\y) circle (0.1);
    }
    \node [label] (a) at (1,-2) {$a$};
    \node [label] (b) at (4,-2) {$b$};
    \node [label] (ap) at (7,4) {$a'$};
    \node [label] (bp) at (10,4) {$b'$};
    \node (nu1) at (2.5,-3) {$\sigma_1$};
    \node (nu2) at (8.5,3) {$\sigma_2$};
    \draw [line width=1pt,->,>=latex']
    (1,4) .. controls +(0,10) and +(0,12) .. (7,7);
    \draw [line width=1pt,->,>=latex']
    (4,1) .. controls +(0,10) and +(0,12) .. (10,10);
    \end{tikzpicture}
    \;
    \begin{tikzpicture}
      [
        scale=0.2,
        label/.style={anchor=base}
      ]
      \draw[step=1cm,black!30,ultra thin,fill=black!10] (-0.2,-0.2) grid (5.2,5.2);
      \foreach \x/\y in {0/5,1/4,2/3,3/2,4/1,5/0} {
          \draw [fill=black] (\x,\y) circle (0.1);
      }
      \draw[step=1cm,black!30,ultra thin,fill=black!10] (5.8,5.8) grid (11.2,11.2);
      \foreach \x/\y in {6/6,7/7,8/8,9/9,10/10,11/11} {
          \draw [fill=black] (\x,\y) circle (0.1);
      }
      \node [label] (a) at (1,-2) {$a'$};
      \node [label] (b) at (4,-2) {$b'$};
      \node [label] (ap) at (7,4) {$a$};
      \node [label] (bp) at (10,4) {$b$};
      \node (nu1) at (2.5,-3) {$\sigma_1$};
      \node (nu2) at (8.5,3) {$\sigma_2$};
      \draw [line width=1pt,<-,>=latex']
      (1,4) .. controls +(0,10) and +(0,12) .. (7,7);
      \draw [line width=1pt,<-,>=latex']
      (4,1) .. controls +(0,10) and +(0,12) .. (10,10);
      \end{tikzpicture}
    \label{subfig:decreasing before and below increasing}
  }
  \caption{\label{fig:subfig:no (nu_1, nu_2)-edge}%
    Illustration of Lemma~\ref{lemma:at most one edge monotone}.
  }
\end{figure}
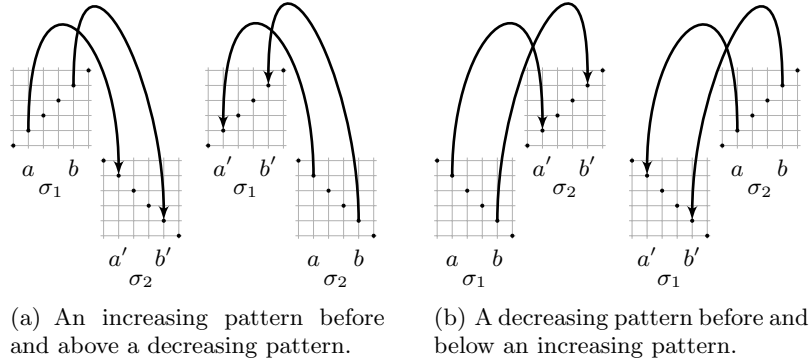

Before proving hardness, we give an easy lemma that will prove extremely
useful for simplifying the proof of upcoming
Proposition~\ref{proposition:hardness}.

\begin{lemma}
  \label{lemma:at most one edge monotone}
  Let $\pi = \pi_1 \, \sigma_1 \, \pi_2 \, \sigma_2 \, \pi_3$
  be a permutation with $|\sigma_1| \geq 2$ and $|\sigma_2| \geq 2$,
  and $\mathcal{M}$ be an oriented perfect matching on $\pi$
  satisfying~$\mathbf{P_1}$ and $\mathbf{P_2}$.
  The following assertions hold:
  \begin{enumerate}
    \item
    If $\sigma_1$ is increasing, $\sigma_2$ is decreasing, and
    $\sigma_1$ is above $\sigma_2$
    (see Fig.~\ref{subfig:increasing before and above decreasing}),
    then there is at most one arc between $\sigma_1$ and
    $\sigma_2$ in $\mathcal{M}$ (this arc can be a
    $(\sigma_1, \sigma_2)$-arc or a $(\sigma_2, \sigma_1)$-arc).
    \item
    If $\sigma_1$ is decreasing, $\sigma_2$ is increasing, and $\sigma_1$
    is below $\sigma_2$
    (see Fig.~\ref{subfig:decreasing before and below increasing}),
    then there is at most one arc between $\sigma_1$ and
    $\sigma_2$ in $\mathcal{M}$ (this arc can be a
    $(\sigma_1, \sigma_2)$-arc or a $(\sigma_2, \sigma_1)$-arc).
  \end{enumerate}
\end{lemma}

\begin{proposition}
  \label{proposition:hardness}
  Deciding whether a permutation is a square is \NPC.
\end{proposition}

\begin{proof}[of Proposition~\ref{proposition:hardness}]
  The problem is certainly in $\NPclass$. We propose a reduction from
  the pattern involvement problem which is known to be
  \NPC~\cite{Bose:Buss:Lubiw:1998}: Given two permutations $\pi$ and
  $\sigma$, decide whether $\sigma$ occurs in $\pi$ (as an
  order-isomorphic pattern).

    \begin{figure}[t!]
    \centering
    \begin{tikzpicture}[
      scale=.9,
      >=stealth',
      shorten >=1pt,
      main node/.style={align=center},
      cell/.style={draw,ultra thick,fill=black!5},
      structure link/.style={line width=1.5pt},
      pattern link/.style={line width=1.5pt,black!20},
      monotone/.style={->,thick}
      ]
      \draw [pattern link,->,>=latex'] (8.5,6) .. controls +(0,1) and +(0,2) .. (10.4,5);
      \draw [pattern link,->,>=latex'] (10.6,5) .. controls +(0,1) and +(0,3) .. (11.5,4);
      \draw [structure link,->,>=latex'] (0.5,11) .. controls +(0,2) and +(0,1) .. (2.5,12);
      \draw [structure link,->,>=latex'] (1.5,2) .. controls +(0,1) and +(0,3) .. (6.5,1);
      \draw [structure link,->,>=latex'] (3.5,9) .. controls +(0,3) and +(0,1.5) .. (7.5,10);
      \draw [structure link,->,>=latex'] (5.5,8) .. controls +(0,1) and +(0,7) .. (9.5,3);
      \draw [structure link,->,>=latex'] (4.5,6) .. controls +(0,-2) and +(0,-1) .. (8.5,5);
      \draw [cell] (0,10) -- (1,10) -- (1,11) -- (0,11) -- cycle;
      \draw [monotone] (0,10) -- ++(1,1) node [midway,fill=white,fill=black!5] {$\nu_1$};
      \draw [cell] (2,11) -- (3,11) -- (3,12) -- (2,12) -- cycle;
      \draw [monotone] (2,11) -- ++(1,1) node [midway,fill=white,fill=black!5] {$\nu'_1$};
      \draw [cell] (1,1) -- (2,1) -- (2,2) -- (1,2) -- cycle;
      \draw [monotone] (1,2) -- ++(1,-1) node [midway,fill=white,fill=black!5] {$\nu_2$};
      \draw [cell] (6,0) -- (7,0) -- (7,1) -- (6,1) -- cycle;
      \draw [monotone] (6,1) -- ++(1,-1) node [midway,fill=white,fill=black!5] {$\nu'_2$};
      \draw [cell] (3,8) -- (4,8) -- (4,9) -- (3,9) -- cycle;
      \draw [monotone] (3,8) -- ++(1,1) node [midway,fill=white,fill=black!5] {$\nu_3$};
      \draw [cell] (7,9) -- (8,9) -- (8,10) -- (7,10) -- cycle;
      \draw [monotone] (7,9) -- ++(1,1) node [midway,fill=white,fill=black!5] {$\nu'_3$};
      \draw [cell] (5,7) -- (6,7) -- (6,8) -- (5,8) -- cycle;
      \draw [monotone] (5,8) -- ++(1,-1) node [midway,fill=white,fill=black!5] {$\nu_4$};
      \draw [cell] (9,2) -- (10,2) -- (10,3) -- (9,3) -- cycle;
      \draw [monotone] (9,3) -- ++(1,-1) node [midway,fill=white,fill=black!5] {$\nu'_4$};
      \draw [cell] (4,6) -- (5,6) -- (5,7) -- (4,7) -- cycle;
      \node [main node] (sigma1) at (4.5,6.5) {$\sigma'$};
      \draw [cell] (8,5) -- (9,5) -- (9,6) -- (8,6) -- cycle;
      \node [main node] (pi1) at (8.5,5.5) {$\pi'$};
      \draw [cell] (10,4) -- (11,4) -- (11,5) -- (10,5) -- cycle;
      \node [main node] (sigma2) at (10.5,4.5) {$\pi''$};
      \draw [cell] (11,3) -- (12,3) -- (12,4) -- (11,4) -- cycle;
      \node [main node] (pi2) at (11.5,3.5) {$\sigma''$};
      \draw[step=1cm,black!30,ultra thin,fill=black!10] (0,0) grid (12,12);
      \foreach \y/\N in {0.5/N_2,1.5/N_2,2.5/N_4,3.5/k,4.5/n,5.5/n+2,6.5/k+2,7.5/N_4,8.5/N_3,9.5/N_3,10.5/N_1,11.4/N_1} {
          \node at (-0.5,\y) (R\y) {$\N$};
      }
      \foreach \x/\N in
      {0.5/N_1,1.5/N_2,2.5/N_1,3.5/N_3,4.5/k+2,5.5/N_4,6.5/N_2,7.5/N_3,8.5/n+2,9.5/N_4,10.5/n,11.5/k} {
          \node at (\x,-0.5) (C\x) {$\N$};
      }
    \end{tikzpicture}
    \caption{\label{fig:reduction}%
      Schematic representation of the permutation $\mu$ used
      in Proposition~\ref{proposition:hardness}.
      Black arcs denote the presence of at least one arc between two bunches of
      positions in $\mu$.
      Grey arcs denote edges that are only considered in the forward direction of the
      proof.
    }%
  \end{figure}
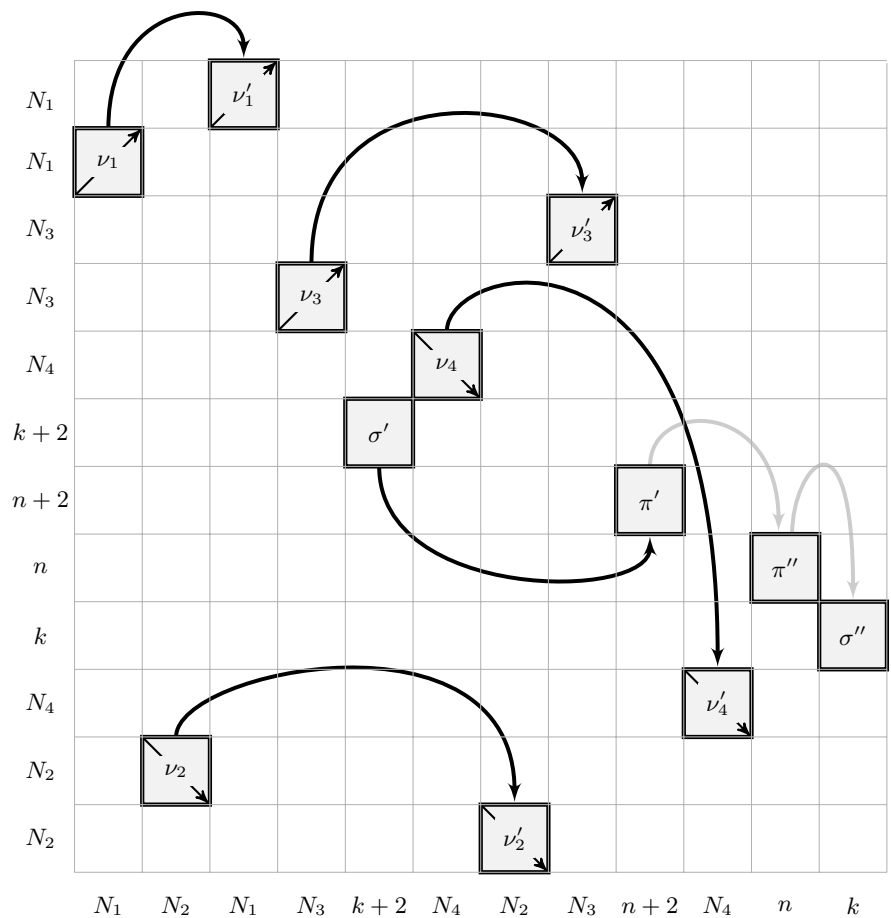

  Let $\pi \in S_n$ and $\sigma \in S_k$ be two arbitrary permutations.
  Define
  \begin{align*}
  N_4 &= 2(2n + 2k + 4) + 1  = 4n + 4k + 9 \\
  N_3 &= 2(2N_4 + 2n + 2k + 4) + 1 = 20n + 20k + 45 \\
  N_2 &= 2(2N_3 + 2N_4 + 2n + 2k + 4) + 1 = 100n + 100k + 225 \\
  N_1 &= 2(2N_2 + 2N_3 + 2N_4 + 2n + 2k + 4) + 1 = 1000n + 1000k  + 1325\text{.}
  \end{align*}
  Notice that $N_1, N_2, N_3$ and $N_4$ are polynomial in $n$.
  The crucial properties are that
  (i) $N_1, N_2, N_3$ and $N_4$ are odd integers
  and
  (ii) $N_i > \left(\sum_{i < j \leq 4} 2N_j\right) + 2n + 2k + k$
  for every $1 \leq i \leq 4$.

  We now turn to defining various gadgets (sequences of integers)
  that act as building blocks in our construction of a new permutation $\mu$:
  \begin{align*}
  \sigma'  &= ((k+1) \; \sigma \; (k+2)) \; [2N_2 + N_4 + 2n + k + 2] \\
  \pi'     &= ((n+1) \; \pi \; (n+2)) \; [2N_2 + N_4 + n + k] \\
  \sigma'' &= \sigma \; [2N_2 + N_4] \\
  \pi''    &= \pi \; [2N_2 + N_4 + k] \\
  \nu_1    &= \nearrow_{N_1} \; [2N_2 + 2N_3 + 2N_4 + 2n + 2k + 4] \\
  \nu'_1   &= \nearrow_{N_1} \; [N_1 + 2N_2 + 2N_3 + 2N_4 + 2n + 2k + 4] \\
  \nu_2    &= \nearrow_{N_2} \; [N_2] \\
  \nu'_2   &= \searrow_{N_2} \\
  \nu_3    &= \nearrow_{N_3} \; [2N_2 + 2N_4 + 2n + 2k + 4] \\
  \nu'_3   &= \nearrow_{N_3} \; [2N_2 + N_3 + 2N_4 + 2n + 2k + 4] \\
  \nu_4    &= \searrow_{N_4} \; [2N_2 + N_4 + 2n + 2k + 4] \\
  \nu'_4   &= \searrow_{N_4} \; [2N_2]\text{.}
  \end{align*}
  We are now in position to define our target permutation $\mu$
  (see Fig.~\ref{fig:reduction} for an illustration):
  $$
  \mu
  =
  \nu_1 \; \nu_2 \; \nu'_1 \; \nu_3 \; \sigma' \; \nu_4 \; \nu'_2 \; \nu'_3 \; \pi' \; \nu'_4 \; \pi'' \; \sigma''
  \text{.}
  $$

  It is immediate that $\mu$ can be constructed in polynomial-time in $n$ and $k$.
  It can be shown
  that $\sigma$ occurs in $\pi$ if and only if
  there exists an oriented perfect matching $\mathcal{M}$ on $\mu$
  satisfying $\mathbf{P_1}$ and $\mathbf{P_2}$.
  \qed
\end{proof}

\section{Conclusion}
\label{section:Conclusion}

There are a number of further directions of investigation in this
general subject. They cover several areas: algorithmic, combinatorics,
and algebra. Let us mention several -~not necessarily new~- open
problems that are, in our opinion, the most interesting. How many
permutations of $S_{2n}$ are squares? How many $(213,231)$-avoiding
permutations of $S_{2n}$ are squares? (Equivalently, by
Proposition~\ref{prop:bijection_binary_to_permutations_squares},
how many binary strings of length $2n$ are squares; see also Problem~4
in \cite{Henshall:Rampersad:Shallit:2011})? How hard is the problem of
deciding whether a $(213,231)$-avoiding permutation is a square
(Problem~4 in \cite{Henshall:Rampersad:Shallit:2011},
see also \cite{Buss:Soltys:2014,Rizzi:Vialette:CSR:2013})?
Given two permutations $\pi$ and~$\sigma$, how hard is the problem of
deciding whether $\sigma$ is a square root of~$\pi$?
As for algebra, one can ask for a complete algebraic study of
$\QQ[S]$ as a graded associative algebra for the shuffle  product
$\SHUFFLE$. Describing a generating family for $\QQ[S]$, defining
multiplicative bases of $\QQ[S]$, and determining whether $\QQ[S]$ is
free as an associative algebra are worthwhile questions.



\bibliographystyle{plain}
\bibliography{biblio}



\end{document}